\documentclass[11pt]{article}
\usepackage{graphicx}
\usepackage{subfigure}
\usepackage{float}
\usepackage{amsfonts,amsbsy,amssymb,amsmath,amsthm,amsfonts,mathtools}
\usepackage{verbatim}
\usepackage{parskip}
\usepackage[T1]{fontenc}
\usepackage{multicol}
\usepackage[onehalfspacing]{setspace}
\usepackage{color}
\usepackage[autostyle]{csquotes}
\usepackage{pgf}
\usepackage{hyperref}
\usepackage{thmtools}
\usepackage{caption}
\usepackage{soul}
\usepackage{pdflscape}
\usepackage{relsize}
 
\declaretheoremstyle[%
  spaceabove=-6pt,%
  spacebelow=6pt,%
  headfont=\normalfont\itshape,%
  postheadspace=1em,%
  qed=\qedsymbol%
]{mystyle} 

\hypersetup{
    colorlinks=true,
    linkcolor=blue,
    filecolor=blue,      
    urlcolor=blue,
    citecolor=blue,
    pdftitle={Overleaf Example},
    pdfpagemode=FullScreen,
    }

\newtheorem{thm}{Theorem}[section]
\newtheorem{prop}{Proposition}[section]
\newtheorem{cor}[thm]{Corollary}
\newtheorem{defn}[thm]{Definition}

\newtheorem{lemma}[thm]{Lemma}

\newcommand{\FF}{{\mathbb{F}}}

\newcommand{\rk}{\texttt{rank}}

\title{Hermitian LCD $2$-Quasi Abelian Codes over Finite Chain Rings}

\author{Sanjit Bhowmick$^{1}$ and Kuntal Deka$^{1}$ 
\footnote{
$^{1}$Department of Electronics and Electrical Engineering, Indian Institute of Technology Guwahati,\\
 Assam, 781039, India. Email:sanjitbhowmick@rnd.iitg.ac.in ;kuntaldeka@iitg.ac.in\\
  }
}

\begin{document}
\maketitle
\begin{abstract} 
This paper introduces a class of Hermitian LCD $2$-quasi-abelian codes over finite fields and presents a comprehensive enumeration of these codes in which relative minimum weights are small. We show that such codes are asymptotically good over finite fields. Furthermore, we extend our analysis to finite chain rings by characterizing $2$-quasi-abelian codes in this setting and proving the existence of asymptotically good Hermitian LCD $2$-quasi-abelian codes over finite chain rings as well.
\end{abstract}

\noindent\textbf{2020 AMS Classification Code:} 94B05; 94B65.

\section{Introduction}\label{sec:intr}
Linear codes with complementary duals (LCD codes) are a class of linear codes that intersect trivially with their respective dual codes, as introduced by Massey \cite{Mas92}. He provided an algebraic characterization of LCD codes and showed that there are asymptotically good LCD codes. He also showed that binary LCD codes provide an optimum linear coding solution for the two-user binary adder
channel. Sendrier showed that LCD codes achieve the Gilbert-Varshamov bound (\cite{Sendrier04}). Apart from this, these codes are significant in both theoretical and applied contexts, with notable importance in countermeasures against passive and active side-channel attacks in embedded cryptosystems \cite{BCC14,CG16}. In addition to cryptography, LCD codes also find applications in communication systems, consumer electronics, and data storage. This wide application has led to extensive research focused on methods for constructing these codes \cite{CMT18,CMTQ18,CMTQ19}. 

The family of primitive BCH codes over finite fields is not asymptotically good. The question of whether cyclic codes are asymptotically good remains an open problem (see \cite{MW06}). In contrast, it has long been established that $2$-quasi-cyclic codes are asymptotically good, as shown in \cite{CPW1969,Chepyzhov1992,M07}. In \cite{Fan16}, quasi-cyclic codes of a particular index were introduced, and these codes have been proven to be asymptotically good. Recently, Fan and Lin \cite{FL15} showed the existence of many asymptotically good quasi-abelian codes that achieve the Gilbert-Varshamov bound. Dihedral group codes, which are non-abelian, closely resemble cyclic group codes. Bazzi and Mitter \cite{BM2007} proved that binary dihedral group codes are asymptotically good. Additionally, in \cite{FL21}, Fan and Lin further showed that dihedral group codes over any finite field, with desirable mathematical properties, are also asymptotically good.

In the 1994s, it was established that many binary non-linear codes can be represented as Gray images of linear codes over the ring $\mathbb{Z}_4$ (see \cite{Hammons94}). This finding led to a growing interest in the study of linear codes over finite commutative chain rings \cite{Hon20,G.Norton2000,Norton2000}. In recent years, substantial research has focused on the examination of LCD codes over finite commutative rings \cite{BDM20,BBB24,BTP24,Liu2015,Liu2020}. 
Recently, Zhang et al. established that Hermitian self-dual $2$-quasi-abelian codes over finite fields are asymptotically good (see \cite{ZLQR2024}). On the other hand, Zhang et al. showed that Euclidean LCD $2$-quasi-abelian codes over finite fields are asymptotically good (see \cite{Zh2024}). So, it is a natural question to ask whether there exist asymptotically good Hermitian LCD $2$-quasi-abelian codes over finite commutative chain rings. Motivated by these previous works and the question, we study Hermitian LCD $2$-quasi-abelian codes over finite commutative chain rings.  First, we build a special type of Hermitian LCD $2$-quasi-abelian codes over any finite field (Theorem~\ref{th-3.1}). Then, we count how many such codes exist (Theorem~\ref{th-3.7}) and estimate the proportion with low relative minimum weights (Theorem~\ref{th-3.11}). These two results lead to the conclusion of the paper (Theorem~\ref{th-3.17}). Finally, we show that over any finite chain ring, there are LCD $2$-quasi-abelian codes that are asymptotically good (Theorem~\ref{th-4.12}).

The present work is organized as follows. In Section~\ref{sec:2}, we review the concept of Galois extensions of finite chain rings, along with some necessary facts about abelian codes over finite chain rings and their Hermitian duals. In Section~\ref{ff-3} discusses Hermitian LCD $2$-quasi-abelian codes over finite fields. Additionally, we construct and count a class of Hermitian LCD codes over finite fields. Further, we show that asymptotically good Hermitian LCD $2$-quasi-abelian codes exist over finite fields. In Section~\ref{lcd}, we provide a characterization of Hermitian LCD $2$-quasi-abelian codes over finite chain rings. Finally, we establish that asymptotically good Hermitian LCD $2$-quasi-abelian codes exist over finite chain rings.
 
\section{Some preliminaries}\label{sec:2} 
Throughout this paper, we denote $\mathcal{S}$ and $\mathcal{R}$ as two finite commutative chain rings (a chain ring means the lattice of all its ideals forms a chain) and $\FF_q$ is a finite field, for the prime power $q$. Let $G$ be a finite abelian group of odd order $n\geq 7$ with $\gcd\{n,q\}=1$. We further assume that the chain ring $\mathcal{S}$ has maximal ideal $\textbf{m}$ with a nilpotency index $s$ (i.e., $\textbf{m}^s=0$ but $\textbf{m}^{s-1}\neq 0$). We say that $\mathcal{S}$ is a ring extension of $\mathcal{R}$, denoted as $\mathcal{S}|\mathcal{R}$, if $\mathcal{R}$ is a subring of $\mathcal{S}$ and $1_{\mathcal{R}}=1_{\mathcal{S}}$. The extension $\mathcal{S}|\mathcal{R}$ is a Galois extension of degree $2$ if $\mathcal{S}$ is isomorphic to $\frac{\mathcal{R}[x]}{(f(x))}$, where $f$ is a monic basic polynomial of degree $2$ over $\mathcal{R}$. The group $\textit{Aut}_{\mathcal{R}}(\mathcal{S})$, which corresponds to this Galois extension, consists of all ring automorphisms $\sigma$ of $\mathcal{S}$ that act as the identity on $\mathcal{R}$. We further assume that the residue fields of the rings $\mathcal{S}$ and $\mathcal{R}$ are $\mathbb{F}=\mathbb{F}_{q^2}$ and $\mathbb{F}_q$, respectively. According to \cite[Theorem XV.2]{McDonal1974}, we have $\textit{Aut}_{\mathbb{F}_q}(\mathbb{F})\equiv\textit{Aut}_{\mathcal{R}}(\mathcal{S})$. The ring $\mathcal{S}$ can be regarded as a free $\mathcal{R}$-module of rank $2$ and $\rk_{\mathcal{R}}(\mathcal{S})=\mid \textit{Aut}_{\mathcal{R}}(\mathcal{S})\mid$.  
Let $\mathcal{S}|\mathcal{R}$ be a Galois extension of finite chain rings of degree $2$, and let $\sigma$ denote a generator of $\textit{Aut}_{\mathcal{R}}(\mathcal{S})$ of order $2$. We introduce a non-degenerate $\sigma$-sesquilinear form 
\begin{eqnarray*}
&&\langle~,~ \rangle~:~\mathcal{S}^n\times \mathcal{S}^n\rightarrow\mathcal{S}~\text{is defined as}~\\
&& \langle \textbf{u},\textbf{v}\rangle_{\sigma}=\sum_{j=1}^{n}u_j\sigma(v_j), ~\text{where},~\textbf{u}=(u_1,u_2,\ldots,u_n)~\text{and}~\textbf{v}=(v_1,v_2,\ldots,v_n). 
\end{eqnarray*}
We say a linear code $C$ is an $\mathcal{S}$-submodule of $\mathcal{S}^n$ and the corresponding $\sigma$-dual of $C$ is $$C^{\perp_{\sigma}}=\{\textbf{u}\in\mathcal{S}^n~|~\langle \textbf{u}, \textbf{c} \rangle_{\sigma}=0~\forall~\textbf{c}=(c_1,c_2,\ldots,c_n)\in C\}.$$
A linear code over $\mathcal{S}$ is called Hermitian linear complementary dual (H-LCD) if it meets with its $\sigma$-duals trivially. 

 The group ring of $G$ over $\mathcal{S}$, denoted by $\mathcal{S}G$, consists of all finite $\mathcal{S}$-linear combinations of elements of $G$, i.e.,
$$\mathcal{S}G=\left\{\sum_{g\in G}a_gg~|~a_g\in\mathcal{S}\right\}.$$
The scalar, additive and multiplication operators of $\mathcal{S}G$ are defined by 
$k\textbf{a}=\sum_{g\in G}(ka_g)g$, $\textbf{a}+\textbf{b}=\sum_{g\in G}(a_g+b_g)g$, and $\textbf{ab}=\sum_{g,h\in G}a_gb_{g^{-1}h}h$, where $k\in\mathcal{S}$, $\textbf{a}=\sum_{g\in G}a_gg$ in $\mathcal{S}G$, and $\textbf{b}=\sum_{g\in G}b_g g$ in $\mathcal{S}G$. Then $\mathcal{S}G$ is a group ring with the identity $1=1_{\mathcal{S}}1_{G}$, where $1_{\mathcal{S}}$ and $1_{\mathcal{G}}$ denote the identity of the ring $\mathcal{S}$ and the group $G$, respectively.
Note that $\mathcal{S}G$ is a finite commutative ring as $\mathcal{S}$ and $G$ are finite. So, by \cite[Theorem~VI.2 and Proof of Theorem~VI.2]{McDonal1974}, $\mathcal{S}G$ can be decomposed uniquely as
\begin{equation}\label{sg-1}
 \mathcal{S}G=\mathcal{S}Ge'_0\oplus\mathcal{S}Ge'_1\oplus\cdots\oplus\mathcal{S}Ge'_t,  
\end{equation}
where, each $e'_i$ is idempotent, $e'_ie'_j=0$, for $1\leq i\neq j\leq t$, and $\sum_{i=1}^te'_i=1$.
Every element $\textbf{a}=\sum_{g\in G}a_gg$ of $\mathcal{S}G$ can be viewed as a word $(a_g)_{g\in G}$ of length $n$ over $\mathcal{S}$, so the Hamming weight of $\textbf{a}$ denoted as $wt_{H}(\textbf{a})$, is defined by $$wt_{H}(\textbf{a})=wt_{H}((a_g)_{g\in G})=\mid\{a_g\in\mathcal{S}~|~a_g\neq 0~\forall~g\in G\} \mid.$$
Now, we introduce Hermitian inner product $\langle ~,~\rangle_{H}$ as follows 
\begin{eqnarray*}
    &&\langle ~,~\rangle_{H}~:~\mathcal{S}G\times \mathcal{S}G\rightarrow \mathcal{S}~\text{is defined by}~\\
    &&\left\langle \sum_{g\in G}a_gg,\sum_{h\in G}b_hh\right\rangle_{H}=\sum_{g\in G}a_g\sigma(b_g), ~\text{for all}~ \sum_{g\in G}a_gg,\sum_{h\in G}b_hh\in \mathcal{S}G.
\end{eqnarray*}
 We call a code $\mathfrak{C}$ of $\mathcal{S}G$ is abelian if $\mathfrak{C}$ is an ideal of $\mathcal{S}G$. With the above Hermitian inner product, the corresponding Hermitian dual of $\mathfrak{C}$ of $\mathcal{S}G$ is defined as $$\mathfrak{C}^{\perp_H}=\{\textbf{a}=\sum_{g\in G}a_gg\in \mathcal{S}G~|~\langle \textbf{a},\textbf{c} \rangle_{H}=0~\text{for all}~\textbf{c}\in \mathfrak{C}\}.$$
We say that an abelian code $\mathfrak{C}$ of $\mathcal{S}G$ is Hermitian LCD if $\mathfrak{C}\cap \mathfrak{C}^{\perp_H}=\{0\}$. Next, we consider a mapping $\phi~:~\mathcal{S}G\rightarrow\mathcal{S}$ by $\sum_{g\in G}a_gg\mapsto a_{1_G}$. It is clear that $\phi$ is a surjective ring homomorphism. Furthermore, it is easy to observe
that $\langle \textbf{a}, \textbf{b} \rangle_{H}=\phi\left(\textbf{a}\widehat{\sigma}(\textbf{b})\right)$, where the mapping $\widehat{\sigma}~:~\mathcal{S}G\rightarrow \mathcal{S}G$ is defined by $\textbf{a}=\sum_{g\in G}a_gg\mapsto \widehat{\sigma}(\textbf{a})=\sum_{g\in G}\sigma(a_g)g^{-1}$. One can immediately verify that $\widehat{\sigma}$ is a ring automorphism on $\mathcal{S}G$.
There is a natural subjective ring homomorphism $\mathcal{S}$ to $\mathbb{F}$, i.e., $\pi : \mathcal{S}\rightarrow\mathbb{F}$ is defined by $r\mapsto \pi(r)=r+\textbf{m}$, for any $r\in \mathcal{S}$. Naturally, we extend this map $\pi :  \mathcal{S}G\rightarrow\mathbb{F}G$ by $r\mapsto \pi(r)$, for any $r\in \mathcal{S}G$.

By Maschke's Theorem (\cite{AB1995}) and under the condition $\gcd\{n,q\}=1$, $\mathbb{F}G$ can be uniquely decomposed as
\begin{equation}
  \mathbb{F}G=\mathbb{F}Ge_0\oplus \mathbb{F}Ge_1\oplus \mathbb{F}Ge_2\oplus\cdots\oplus \mathbb{F}Ge_r\oplus \mathbb{F}Ge_{r+1}\oplus\cdots\oplus \mathbb{F}Ge_{r+s},  
\end{equation}\label{fg-2}
where each $e_i$ is primitive idempotent element and $e_ie_j=0$ for $i\neq j$ and $1=\sum_{1}^{r+s}e_i$. One can verify that $\mathbb{F}Ge_i$ is a vector space over $\mathbb{F}$ and then write
$$\mu_q(n)=\min\{\dim_{\mathbb{F}}\mathbb{F}Ge~\big|~e\in E\setminus\{e_0\}\},$$ for more details to see \cite[Lemma II.2]{Fan22}.
Note that the restriction map on $\mathbb{F}G$, i.e., $\tau:=\widehat{\sigma}\mid_{\mathbb{F}G}~:~\mathbb{F}G\rightarrow\mathbb{F}G$, is an automorphism on $\mathbb{F}G$ and is defined by $\alpha=\sum_{g\in G}\alpha_gg\mapsto \alpha^\tau=\sum_{g\in G}\alpha_g^qg^{-1}$.
In addition, we fix the notation as
\begin{itemize}
 \item $E=\{e_0\}\cup \{e_1,e_2,\ldots,e_r\}\cup\{e_{r+1},e_{r+1}^{\tau},\ldots,e_{r+s},e_{r+s}^{\tau}\}$, where $e_0^\tau=e_0$, $e_i^\tau=e_i$, for $i=1,2,\ldots,r$, and $e_{r+j}^\tau\neq e_{r+j}$, for $j=1,2,\ldots,s$.
\item $E^*=E\setminus\{e_0\}=\{e_1,e_2,\ldots,e_r\}\cup\{e_{r+1},e_{r+1}^{\tau},\ldots,e_{r+s},e_{r+s}^{\tau}\}$.
\item $\widehat{e}_{r+j}=e_{r+j}+e_{r+j}^\tau$, where $j=1,2,\ldots,s$.
\item $\widehat{E}=\{e_0,e_1,\ldots,e_r,\widehat{e}_{r+1},\ldots,\widehat{e}_{r+s}\}$.
\item $\widehat{E}^\dagger=\widehat{E}\setminus\{e_0\}=\{e_1,\ldots,e_r,\widehat{e}_{r+1},\ldots,\widehat{e}_{r+s}\}$.
\item $\mathcal{A}_{i}=\mathbb{F}Ge_i$, for $i=0,1,\ldots,r+s$.
\item $\widehat{\mathcal{A}}_{r+j}=\mathbb{F}Ge_i+\mathbb{F}G\widehat{e}_{r+j}$, for $j=1,\ldots,s$.
\item $\mathcal{I}^\tau=\{a\in \mathcal{I}~|~a^\tau=a\}$, where $\mathcal{I}$ is an ideal of $\mathbb{F}G$.
\end{itemize}
Furthermore, we notice that the restriction of $\tau$ on $\mathcal{A}_{r+j}$ induces a map $\mathcal{A}_{r+j}=\mathbb{F}Ge_{r+j}\rightarrow\mathcal{A}_{r+j}^\tau=\mathbb{F}Ge_{r+j}^\tau$ as $a\mapsto a^\tau$ is an $\mathbb{F}$-linear isomorphism. It is immediately verified that $\mathcal{A}_{r+j}+\mathcal{A}_{r+j}^\tau$ is invariant under the automorphism $\tau$. Now, we provide a lemma which found in \cite{ZLQR2024}.
\begin{lemma}\cite[Lemma~2.2]{ZLQR2024}\label{lm-2.1}
We are keeping the notation as mentioned above. Then the following statements are hold
\begin{enumerate}
\item[a)] $\left( \mathbb{F}Ge_0\right)^\tau=\mathbb{F}_qe_0$; 
\item[b)] $\left(\mathbb{F}G\right)^\tau$ is a vector space over $\mathbb{F}_q$ and $\left(\mathbb{F}G\right)^\tau$ can be decomposed uniquely as $$\left(\mathbb{F}G\right)^\tau=\mathbb{F}_qe_0\bigoplus\left(\oplus_{i=1}^r(\mathbb{F}Ge_i)^\tau\right) \bigoplus\left(\oplus_{j=1}^s(\mathbb{F}G\widehat{e}_{r+j})^\tau\right);$$
\item[c)] $\dim_{\mathbb{F}_q}\left(\mathbb{F}G\right)^\tau=n$;
\item[d)] $(\mathbb{F}G\widehat{e}_{r+j})^\tau=\{a+a^\tau~|~a\in\mathbb{F}Ge_{r+j}\}$ and $\dim_{\mathbb{F}_q}\left(\mathbb{F}G\widehat{e}_{r+j}\right)^\tau=\dim_{\mathbb{F}_q}\mathbb{F}Ge_{r+j}$, for $j=1,\ldots,s$; 
\item[e)]  $\dim_{\mathbb{F}_q}\left(\mathbb{F}Ge_{i}\right)^\tau=\frac{1}{2}\dim_{\mathbb{F}_q}\mathbb{F}Ge_{i}$, for $i=1,\ldots,r$.
\end{enumerate}
\end{lemma}
Further, we denote 
\begin{align*}
    &k_e=\dim_{\mathbb{F}}\mathbb{F}Ge~\text{for all}~e\in \widehat{E}\setminus\{e_0\},&\\
    &k_i=k_{e_i}, ~\text{for}~i=1,2,\ldots,r;&\\
    &k_{r+j}=k_{\widehat{e}_{r+j}}~\text{and}~k_{r+j}~\text{is even} ~\text{for}~j=1,2,\ldots,s;&\\
&\dim_{\mathbb{F}}\mathbb{F}Ge_{r+j}=\dim_{\mathbb{F}}\mathbb{F}Ge_{r+j}^\tau~\text{for}~j=1,2,\ldots,s.&\\
\end{align*}
\section{Hermitian LCD $2$-quasi-abelian codes}\label{ff-3}
The main goal of this section is to construct of a class of Hermitian-LCD $2$-quasi-abelian group code over $\mathbb{F}$. Set $$(\mathbb{F}G)^2:=\mathbb{F}G\times \mathbb{F}G=\{(a,b)~|~a,b\in\mathbb{F}G\}$$
is an $\mathbb{F}G$-module. Any $\mathbb{F}G$-submodule $\mathcal{C}$ of $(\mathbb{F}G)^2$ is called $2$-quasi-abelian code over $\mathbb{F}$. We further extend the Hermitian inner product in $(\mathbb{F}G)^2$ naturally. So in a similar way, we define the Hermitian dual code over $\mathbb{F}$ of $\mathbb{F}G$. For any $a,b\in\mathbb{F}G$, we denote $\mathcal{C}_{a,b}$ as follows
$$\mathcal{C}_{a,b}:=\{s(a,b)~|~s\in\mathbb{F}G\}.$$
Clearly, $\mathcal{C}_{a,b}$ is an $\mathbb{F}G$-submodule of $(\mathbb{F}G)^2$ generated by $(a,b)$. Denote, $\mathbb{F}^\times=\mathbb{F}\setminus\{0\}$.
\begin{thm}\label{th-3.1}
For any $a,b\in\mathbb{F}G$, and if $aa^\tau+bb^\tau=\lambda$, where $\lambda\in\mathbb{F}^\times$, then $\mathcal{C}_{a,b}$ is a Hermitian LCD $2$-quasi-abelian code, i.e., $\mathcal{C}_{a,b}\cap\mathcal{C}_{a,b}^{\perp_H}=\{0\}.$   
\end{thm}
\begin{proof}
Let $s(a,b)\in\mathcal{C}_{a,b}\cap\mathcal{C}_{a,b}^{\perp_H}$, this gives $(sa,sb)\in \mathcal{C}_{a,b}^{\perp_H}$, which implies $\langle(sa,sb) ,(s_1a,s_1b)\rangle_H=0$, for all $s_1\in \mathbb{F}G$. It follows that $\langle sa,s_1a\rangle_H+\langle sb,s_1b\rangle_H=0$ for all $s_1\in\mathbb{F}G$. Therefore, ${\phi\mid_{\mathbb{F}G}}(ss_1^\tau(aa^\tau+bb^\tau))=0$, for all $s_1\in\mathbb{F}G$. By hypothesis, ${\phi\mid_{\mathbb{F}G}}(ss_1^\tau)=0$ for all $s_1\in \mathbb{F}G$. This implies $s=0$. Hence, the result follows immediately. 
\end{proof}
\begin{cor}\label{cor-3.1}
For $\beta\in\mathbb{F}G$ and $\lambda\in\mathbb{F}^\times$ such that $\beta\beta^\tau=\lambda-1\in\mathbb{F}^\times$, $\mathcal{C}_{1,\beta}=\{s(1,\beta)~|~s\in\mathbb{F}G\}$ is a Hermitian LCD $2$-quasi-abelian code with the rate $1/2$.  \end{cor}
\begin{proof}
By applying Theorem~\ref{th-3.1}, take $a=1$ and $b=\beta$, the result follows immediately. For the other part, we define a mapping $\varrho: \mathcal{C}_{1,\beta}\rightarrow\mathbb{F}G$ by $s(1,\beta)\mapsto s$ for all $s\in\mathbb{F}G$. It is easy to verify that $\varrho$ is an $\mathbb{F}G$-module isomorphism between $\mathcal{C}_{1,\beta}$ and $\mathbb{F}G$. Hence, we conclude the desired result.   
\end{proof}
Using the above Corollary~\ref{cor-3.1}, we will construct a class of Hermitian LCD codes with the condition that $\lambda\in\mathbb{F}^\times$ and $\lambda-1\in\mathbb{F}^\times$.

In the following we always denote
\begin{equation}\label{cal D}
\begin{array}{l}
{\cal D}_{\lambda}=\big\{\mathcal{C}_{1,\beta}\,\big|\, \beta\in \mathbb{F}G,\, \beta\beta^\tau=\lambda-1\big\},
\end{array}\end{equation}
which is the set of all Hermitian LCD $2$-quasi-abelian codes, 
see Corollary \ref{cor-3.1}.
We always assume that $\delta$ is a real number such that
$0\le\delta\le 1-q^{-1}$, and set
\begin{equation}\label{D^le}
\textstyle
 {\cal D}^{\le\delta}_{\lambda}=
\big\{ {\cal C}_{1,\beta}\,\big|\, \mathcal{C}_{1,\beta}\in{\cal D},\, 
  \Delta(\mathcal{C}_{1,\beta})=\frac{{ wt_H}(\mathcal{C}_{1,\beta})}{2n}\le\delta\big\}
\end{equation}
and 
\begin{equation}\label{cal D-1}
\begin{array}{l}
\widehat{\cal D}_{\lambda}=\big\{\beta\in \mathbb{F}G\,\big|\, \beta\beta^\tau=\lambda-1\big\}.
\end{array}\end{equation} One can imemediately verify that $\big| {\cal D}_{\lambda} \big|=\big| \widehat{\cal D}_{\lambda} \big|$.
For the estimate of $\widehat{\cal D}_{\lambda}$, we introduce the following lemma.
\begin{lemma}
 For $\beta\in\mathbb{F}G$ and $\lambda\in\mathbb{F}^\times$, then $\beta\beta^\tau=\lambda-1$ if and only if 
 \begin{align*}
    &\qquad \begin{cases*}
        {} \beta\beta^\tau e_0=(\lambda-1)e_0 \\
        {} \beta\beta^\tau e_i=(\lambda-1)e_i  & for $1\leq i \leq r$ \\
        {} \beta\beta^\tau \widehat{e}_{r+j}=(\lambda-1)\widehat{e}_{r+j}  & for $1\leq j \leq s$.
    \end{cases*}
\end{align*}
\end{lemma}
\begin{proof}
 Its proof is a straightforward exercise.   
\end{proof}
Further, we denote the following sets
 \begin{eqnarray}\label{eq-3.5}
     &&\mathcal{I}_{0}=\{\alpha\in\mathbb{F}Ge_0~\big|~\alpha\alpha^\tau=(\lambda-1)e_0\};\\\label{eq-3.6}
     &&\mathcal{I}_{i}=\{\alpha\in\mathbb{F}Ge_i~\big|~\alpha\alpha^\tau=(\lambda-1)e_i\}~\text{for all}~i=1,2,\ldots,r;\\\label{eq-3.7}
     &&\mathcal{I}_{r+j}=\{\alpha\in\mathbb{F}G\widehat{e}_{r+j}~\big|~\alpha\alpha^\tau=(\lambda-1)\widehat{e}_{r+j}\}~\text{for all}~j=1,2,\ldots,s.
 \end{eqnarray}
\begin{lemma}\label{lm-3.2}
Let $\mathcal{I}_{0}$ be a set defined in Equation~\eqref{eq-3.5}. Then $
\big|\mathcal{I}_{0}\big|=
\begin{cases}
q + 1 & \text{if } \lambda-1 \in \mathbb{F}_q^*, \\
0     & \text{if } \lambda-1 \notin \mathbb{F}_q^*.
\end{cases}
$    
\end{lemma}
\begin{proof}
By Equation~\eqref{eq-3.5}, we get that
\begin{align*}
   \big|\mathcal{I}_{0}\big|&=\big|\{\alpha\in\mathbb{F}Ge_0~\big|~\alpha\alpha^\tau=(\lambda-1)e_0\}\big|&\\
   &=\big|\{\nu e_0\in\mathbb{F}Ge_0~\big|~\nu e_0(\nu e_0)^\tau=(\lambda-1)e_0\}\big|&\\
     &=\big|\{\nu\in\mathbb{F}~\big|~\nu\nu^\tau=(\lambda-1)\}\big|.&\\
\end{align*} 
Thus,
$$
\big|\mathcal{I}_{0}\big|=\big|\{\nu\in\mathbb{F}~\big|~\nu\nu^\tau=(\lambda-1)\}\big|=
\begin{cases}
q + 1 & \text{if } \lambda-1 \in \mathbb{F}_q^*, \\
0     & \text{if } \lambda-1 \notin \mathbb{F}_q^*.
\end{cases}.
$$
This follows since the map $\nu \mapsto {\nu}^{q+1}$ from $\mathbb{F}_{q^2}^* \to \mathbb{F}_q^*$ is a surjective group homomorphism with kernel of size $q + 1$.
\end{proof}
\begin{lemma}\label{lm-3.3}
Let $\mathcal{I}_{i}$ be a set defined in Equation~\eqref{eq-3.6}. Then $\big|\mathcal{I}_{i}\big|=q^{k_i}+1$, for $i=1,2,\ldots,r$.    
\end{lemma}
\begin{proof}
By applying Lemma~\ref{lm-2.1}, we conclude that $(\mathbb{F}Ge_i)^\tau$ is a subfield of $\mathbb{F}Ge_i$ with $\dim(\mathbb{F}Ge_i)^\tau=\frac{1}{2}\dim \mathbb{F}Ge_i$. Note that $\tau$ on $\mathbb{F}Ge_i$ is a Galois automorphism of $\mathbb{F}Ge_i$. It means that the order of $\tau$ is equal to $2$. Thus, we get $\alpha^\tau=\alpha^{q^k_i}$ for all $\alpha\in \mathbb{F}Ge_i$, $i=1,2,\ldots,r$. Hence, $\big|\mathcal{I}_{i}\big|=q^{k_i}+1$, for $i=1,2,\ldots,r$ (see \cite[Theorem 7.16 and Corollary 7.17]{ZWan11}). 
\end{proof} 
\begin{lemma}\label{lm-3.4}
Let $\mathcal{I}_{r+j}$ be a set defined in Equation~\eqref{eq-3.7}. Then $\big|\mathcal{I}_{r+j}\big|=q^{k_{r+j}}-1$, for $j=1,2,\ldots,s$.    
\end{lemma}
\begin{proof}
 To prove the result, let us assume that $\alpha=\beta+\gamma$, where $\beta\in\mathbb{F}Ge_{r+j}$ and $\gamma\in\mathbb{F}Ge_{r+j}^\tau$. Then,
 \begin{align*}
     \mathcal{I}_{r+j}&=\{\beta+\gamma\in\mathbb{F}G\widehat{e}_{r+j}~\big|~(\beta+\gamma)(\beta+\gamma)^\tau=(\lambda-1)\widehat{e}_{r+j}\}&\\
     &=\{\beta+\gamma\in\mathbb{F}G\widehat{e}_{r+j}~\big|~\beta\gamma^\tau+\beta^\tau\gamma=(\lambda-1)e_{r+j}+(\lambda-1)e_{r+j}^\tau\}&\\
     &=\{\beta+(\lambda-1)(\beta^{-1})^\tau {e}_{r+j}^\tau~\big|~\beta\in(\mathbb{F}Ge_{r+j})^{\times}\}&\\
 \end{align*}
 Therefore, $\big|\mathcal{I}_{r+j}\big|=q^{k_{r+j}}-1$, for $j=1,2,\ldots,s$, which completes the proof.
\end{proof}
Moreover, we have the following theorem.
\begin{thm}\label{th-3.7}
 The cardinality of $\mathcal{D}_{\lambda}$ is
 \begin{equation}
     \big|\mathcal{D}_{\lambda}\big|=(q+1)\prod_{i=1}^r(q^{k_i}+1)\prod_{j=1}^s(q^{k_{r+j}}-1).
 \end{equation}
\end{thm}
\begin{proof}
It follows immediately from Lemmas~\ref{lm-3.2}, \ref{lm-3.3}, and \ref{lm-3.4}.    
\end{proof}
\begin{lemma}\cite[Lemma~IV.7]{Fan22}\label{lm-4.4}
  For $q\geq 2$, and $l_1,l_2,\ldots,l_m$ such that $l_i\geq \log_qm$, where $i=1,2,\ldots,m$, then
  \begin{enumerate}
  \item[a)] $\prod_{i=1}^m(q^{l_i}-1)\geq q^{\sum_{i=1}^ml_i}-2$;
  \item[b)] $\prod_{i=1}^m(q^{l_i}+1)\leq q^{\sum_{i=1}^ml_i}+2$.
  \end{enumerate}
\end{lemma}
\begin{cor}\label{cor-4.1}
 For $\dfrac{\log_qn}{\mu_q(n)}\leq 1$, we have $\big|\mathcal{D}_{\lambda}\big|\geq q^{n-2}$.   
\end{cor}
\begin{proof}
Note that $n=1+r+s$, we have $\log_q(r+s)<\log_qn\leq \mu_q(n)$. Clearly, $\log_q(r+s)<k_e$ as $\mu_q(n)\leq k_e$. From this and applying Theorem~\ref{th-3.7} and Lemma~\ref{lm-4.4}, we conclude that
$\big|\mathcal{D}_{\lambda}\big|\geq q^{n-2}$.   
\end{proof}
Next, for $a\in\mathbb{F}G$, we denote
\begin{align*}
  \widehat{E}_{a}&=\left\{e~\big|~e\in\widehat{E},~ea\neq 0\right\},\; \;\widehat{E}_{a}^\dagger=\left\{e~\big|~e\in\widehat{E}^\dagger,~ea\neq 0\right\};&\\ 
  L_{a}&=\oplus_{e\in \widehat{E}_{a}^\dagger}\mathbb{F}Ge\leq\mathbb{F}G; \; \; \ell_{a}=\sum_{e\in \widehat{E}_{a}^\dagger}n_e, ~\text{hence},~\dim L_{a}=\ell_{a}. &\\
  &\text{For an integer $\ell$ with $\mu_q(n)\leq \ell<n$, we denote},&\\
  \Omega_{\ell}&=\left\{J\leq \mathbb{F}G~\big|~\text{there exists a subset $\mathcal{S}\subseteq \widehat{E}_{a}^\dagger$ such that $J=\oplus_{e\in\mathcal{S}}\mathbb{F}Ge$, $\dim J=\ell$} \right\};&\\
  &\text{$J\in \Omega_{\ell}$, we denote}~\Tilde{J}=\mathbb{F}Ge_0+J,~\Tilde{J}^*=\left\{a~\big|~a\in\Tilde{J}, L_a=J\right\}, ~\text{hence}, ~\dim\Tilde{J}=\ell+1,~\ell_a=\ell;&\\
  &\mathcal{D}_{a,b}=\left\{\mathcal{C}_{1,\beta}~\big|~(a,b)\in\mathcal{C}_{1,\beta}\right\}.&\\ \end{align*}
\begin{lemma}\label{lm-4.5}
  If $\mathcal{D}_{a,b}\neq\emptyset$, then $\widehat{E}_a=\widehat{E}_b$, and $\mathcal{D}_{a,b}\leq q^{n+3-\ell_a}$.   
\end{lemma}
\begin{proof}
To prove the result, let us assume $\beta=\sum_{e\in\widehat{E}}\beta_e$, where $\beta_e\in\mathbb{F}Ge$ for $e\in\widehat{E}$ and $\mathcal{C}_{1,\beta}\in\mathcal{D}_{a,b}$. It implies $(a,b)=u(1,\beta)$ for some $u\in\mathbb{F}G$. This gives $a=u$ and $b=u\beta$, it follows that $b=a\beta$. Therefore, $eb=ea\beta=ea\beta_e$, for all $e\in\widehat{E}$. Furthermore, we conclude that $\mathcal{C}_{1,\beta}\in\mathcal{D}_{a,b}$ if and only if $\beta_e \beta_e^{\tau}=(\lambda-1)e$ and $e\beta=ea\beta_{e}$, for all $e\in\widehat{E}$. Hence, $\widehat{E}_a=\widehat{E}_b$ and $\big|\mathcal{D}_{a,b} \big|=\prod_{e\in\widehat{E}}\big|\{\beta_e~\big|~\beta_e\beta_e^\tau=(\lambda-1)e,eb=ea\beta_e\}\big|$. 

Case~$1:$ $e=e_0$. By applying Lemma~\ref{lm-3.2}, we have at most $q+1$ choices of $\beta_e$. 

Case~$2:$ $e\in\widehat{E}_{a}^\dagger$. Since $\widehat{E}_{a}^\dagger=\widehat{E}_{b}^\dagger$, hence $ea\neq 0$ and $eb\neq 0$.

Subcase~$2.1:$ $e=e_i$ for some $i$, where $1\leq i\leq r$. We conclude that $\mathbb{F}Ge$ is finite field and $eb=ea\beta_e$, we obtain $\beta_e=(ea)^{-1}(e\beta)$. Thus, we have only one choice of $\beta_e$. 

Subcase~$2.2:$ $e=\widehat{e}_i=e_{r+j}+e_{r+j}^\tau$ for some $j$, where $1\leq j \leq s$. Suppose that $\beta_e=c_e+d_e$, where $c_e\in\mathbb{F}Ge_{r+j}$, and $d_e\in\mathbb{F}Ge_{r+j}^\tau$. Therefore, $\beta_e\beta_e^\tau=(\lambda-1)e$ if and only if $c_ed_e^\tau=(\lambda-1)e_{r+j}$, and $eb=ea\beta_e$ if and only if $be_{r+j}=ae_{r+j}c_e$, $be_{r+j}^\tau=ae_{r+j}^\tau d_e.$ If $ea\neq 0$, we have $e_{r+j}a\neq 0$ or $e_{r+j}^\tau a\neq 0$. In both cases, we conclude that there is only one choice of $\beta_e$. 

Case~$3:$ $e\in\widehat{E}^\dagger-\widehat{E}^\dagger_{a}$. We have $eb=ea=0$. Thus, $eb=ea\beta_e$. Therefore, the number of choices for $\beta_e$ is at most $\text{max}\{q^{k_e}-1,q^{k_e}+1\}=q^{k_e}+1$, which is obvious from Lemmas~\ref{lm-3.3} and \ref{lm-3.4}.
Therefore, we conclude that $$\big|\mathcal{D}_{a,b}\big|\leq (q+1)\prod_{e\in\widehat{E}^\dagger-\widehat{E}^\dagger_{a}}(q^{k_e}+1).$$ Further, we notice that $\sum_{e\in\widehat{E}^\dagger-\widehat{E}^\dagger_{a}}k_e=n-1-\ell_a$. By applying Lemma~\ref{lm-4.5}, we obtain that  $$\big|\mathcal{D}_{a,b}\big|\leq (q+1)\prod_{e\in\widehat{E}^\dagger-\widehat{E}^\dagger_{a}}(q^{k_e}+1)\leq (q+1) q^{\sum_{e\in\widehat{E}^\dagger-\widehat{E}^\dagger_{a}}k_e+2}\leq q^{n+3-\ell_a}.$$ This completes the proof.
\end{proof}
\begin{lemma}\label{lm-3.11}
We have the following relation.
 $$\mathcal{D}_{\lambda}^{\leq \delta}\subseteq \bigcup_{\mu_q\leq \ell<n}\bigcup_{J\in\Omega_{\ell}}\bigcup_{(a,b)\in(\Tilde{J}^*\times\Tilde{J}^*)^{\leq \delta}}\mathcal{D}_{a,b}.$$
\end{lemma}
\begin{proof}
To prove the result, we assume that $\mathcal{C}_{1,\beta}\in\mathcal{D}^{\leq\delta}_{\lambda}$, this gives $0<\frac{wt_H(\mathcal{C}_{1,\beta})}{2n}\leq\delta$ and $\beta\beta^\tau=\lambda-1$. So, there exists a non-zero element $(a,b)$ in $\mathcal{C}_{1,\beta}$ such that $0<wt_{H}(a,b)\leq2n\delta$. Therefore, by Lemma~\ref{lm-4.5}, we obtain $\widehat{E}^\dagger_a=\widehat{E}^\dagger_b$. If possible, assume $\widehat{E}^\dagger_a=\widehat{E}^\dagger_b=\emptyset$, which implies $ea=eb=0$, for all $e\in\widehat{E}^\dagger_a$, it means $a,b$ in $\mathbb{F}Ge_0$. Then there are non zero elements $\nu_a,\nu_b$ of $\mathbb{F}$ such that $a=\nu_ae_0$ and $b=\nu_be_0$. Thus, $wt_H(a,b)=2n$, which is a contradiction as $wt_H(a,b)<2n$. Hence, $\widehat{E}^\dagger_a=\widehat{E}^\dagger_b\neq\emptyset$. Next, let us set $J=\oplus_{e\in\widehat{E}^\dagger_a}\mathbb{F}Ge=\oplus_{e\in\widehat{E}^\dagger_b}\mathbb{F}Ge$, $\Tilde{J}=J+\mathbb{F}Ge_0$, and $\dim_{\mathbb{F}}J=\ell$, then $\mu_q(n)\leq\ell<n$. Therefore, $J\in\Omega_{\ell}$ and $a,b\in\Tilde{J}^*$. Thus, for $(a,d)\in(\Tilde{J}^*\times\Tilde{J}^*)^{\leq \delta}$, then $\mathcal{C}_{1,\beta}\in\mathcal{D}_{a,b}$.
\end{proof}
\begin{lemma}\label{lm-3.12}
 The cardinality of $\Omega_{\ell}$ is
 $$\big|\Omega_{\ell} \big|<n^{\frac{\ell}{\mu_q(n)}}.$$
\end{lemma}
\begin{proof}
The proof of the result follows from the definition of $\Omega_{\ell}$.    
\end{proof}
Next, we define the $q$-entropy function, denote $h_q(\delta)$, by $h_q(\delta)=\delta \log_q(q-1)-\delta \log_q\delta-(1-\delta)\log_q(1-\delta)$. Note that $h_q(\delta)$ increases and is concave in the interval $[0,1-q^{-1}]$ with $h_q(0)=0$ and $h_q(1-q^{-1})=1$.
\begin{lemma}\cite[Lemma~IV.6]{Fan22}\label{lm-3.13}
For an ideal $I$ of $\mathbb{F}G$, then $I\times I\leq(\mathbb{F}G)^2$ and $\big|I\times I\big|^{\leq \delta}\leq q^{h_q(\delta)\dim_{\mathbb{F}}(I\times I)}$.    \end{lemma}
\begin{lemma}\label{lm-3.14}
For $\mu_q(n)\leq \ell<n$ and $I\in \Omega_{\ell}$, we have $\big|\bigcup_{(a,b)\in(\Tilde{J}^*\times\Tilde{J}^*)^{\leq \delta}} \mathcal{D}_{a,b}\big|\leq q^{n+3-2\ell[\frac{1}{2}-h_q(\delta)]+2h_q(\delta)}$.
\end{lemma}
\begin{proof}
By applying Lemma~\ref{lm-3.13}, we obtain 
\begin{eqnarray*}
    \big| \bigcup_{(a,b)\in(\Tilde{J}^*\times\Tilde{J}^*)^{\leq \delta}}\mathcal{D}_{a,b}\big|\leq\sum_{(a,b)\in(\Tilde{J}^*\times\Tilde{J}^*)^{\leq \delta}}\big|  \mathcal{D}_{a,b}\big|&&\\
    \leq \sum_{(a,b)\in(\Tilde{J}^*\times\Tilde{J}^*)}q^{n+3-\ell}~&\text{by Lemma~\ref{lm-4.5}}&\\
    =\big|{(\Tilde{J}^*\times\Tilde{J}^*)^{\leq \delta}} \big|q^{n+3-\ell}&&\\
    \leq q^{n+3-2\ell[\frac{1}{2}-h_q(\delta)]+2h_q(\delta)}. &&\\
\end{eqnarray*} This completes the proof.
\end{proof}
\begin{thm}\label{th-3.11}
 For $\frac{1}{2}-h_q(\delta)-\frac{\log_qn}{\mu_q(n)}>0$, we have $\big|\mathcal{D}^{\delta}_\lambda\big|\leq q^{n+4-2\mu_q(n)\left [\frac{1}{2}-h_q(\delta)-\frac{\log_qn}{\mu_q(n)}\right ]}$.   
\end{thm}
\begin{proof}
At first, we conclude that $\ell\geq \mu_q(n)$ and $\frac{1}{2}-h_q(\delta)-\frac{\log_qn}{\mu_q(n)}>0$. This implies $2h_q(\delta)\leq 1$. By applying Lemmas~\ref{lm-3.11}, \ref{lm-3.12} and \ref{lm-3.14}, the result follows immediately. \end{proof}
\begin{lemma}\cite[Lemma II.6]{FL21}\label{lm-3.15}
There exist infinitely many positive odd integers $n_1$, $n_2$, $\ldots$ coprime with $q$ such that $\lim_{i\rightarrow\infty}\frac{\log_qn_i}{\mu_q(n_i)}=0$, specifically $\mu_q(n_i)\rightarrow \infty$ as $i\rightarrow\infty$. 
\end{lemma}
Now, we are ready to show that a class of Hermitian LCD $2$-quasi-abelian codes over $\mathbb{F}$ is asymptotically good.
\begin{thm}\label{th-3.17}
For an integer $n_i$ (as defined in Lemma~\ref{lm-3.15}) and an abelian group $G_i$ of odd order $n_i$, where $i=1,2,\ldots $, there are Hermitian LCD $2$-quasi-abelian codes $\mathcal{C}_i$ such that the sequence of codes $\mathcal{C}_1$, $\mathcal{C}_2$, $\ldots$ is asymptotically good.    
\end{thm}
\begin{proof}
To prove the results, let us assume that $\mathcal{D}_{\lambda,i}$ (defined as in Equation~\eqref{cal D}) is the set of Hermitian LCD $2$-quasi-abelian codes for $i=1,2,\ldots$. As $\lim_{i\rightarrow\infty}\frac{\log_qn_i}{\mu_q(n_i)}=0$, then we choose $\epsilon>0$ so that $\frac{1}{2}-h_q(\delta)-\frac{\log_qn}{\mu_q(n)}>\epsilon>0$ and $\frac{\log_qn_i}{\mu_q(n_i)}\leq 1$. By applying Corollary~\ref{cor-4.1} and Theorem~\ref{th-3.11}, we get
$$\frac{\mid\mathcal{D}_{\lambda,i}^{\leq \delta}\mid}{\mid\mathcal{D}_{\lambda,i}\mid} \leq \frac{q^{n_i+4-2\mu_q(n_i)[\frac{1}{2}-h_q(\delta)-\frac{\log_qn_i}{\mu_q(n_i)}]}}{q^{n_i-2}}=q^{-2\mu_q(n_i)[\frac{1}{2}-h_q(\delta)-\frac{\log_qn}{\mu_q(n)}]+6}.$$
By taking $\mu_q(n_i)\rightarrow\infty$, we have 
$$\lim_{i\rightarrow\infty}\frac{\mid\mathcal{D}_{\lambda,i}^{\leq \delta}\mid}{\mid\mathcal{D}_{\lambda,i}\mid} \leq\lim_{i\rightarrow\infty} q^{-2\mu_q(n_i)[\frac{1}{2}-h_q(\delta)-\frac{\log_qn_i}{\mu_q(n_i)}]+6}=0.$$ 
Therefore, we can take $\mathcal{C}_i\in \mathcal{D}_{\lambda,i}\setminus\mathcal{D}_{\lambda,i}^{\leq \delta}$ for $i=1,2,\ldots$, there exist Hermitian LCD $2$-quasi-abelian code $\mathcal{C}_i$  with length $n_i$ for $i=1,2,\ldots$ such that $n_i\rightarrow\infty$ and we get
 \begin{enumerate}
 \item[i)] The code $\mathcal{C}_i$ of length $2n_i\rightarrow\infty$;
 \item[ii)] The rate $R(\mathcal{C}_i)=\frac{1}{2}$;
 \item[iii)] The relative minimum distance $\vartriangle(\mathcal{C}_i)>\delta$ for all $i=1,2,\ldots$.
 \end{enumerate} Consequently, we say the sequence $\mathcal{C}_1$, $\mathcal{C}_2$, $\ldots$ of codes are asymptotically good.
\end{proof}

\section{Hermitian LCD $2$-quasi-abelian codes over finite chain rings}\label{lcd}
In this section, we first characterize Hermitian LCD $2$-quasi-abelian codes over finite chain rings. Now, we recall the property of Hermitian LCD $2$-quasi-abelian codes over finite chain rings. Next, $\left(\mathcal{S}G \right)^2 : = \{(a,b) : a,b \in \mathcal{S}G\}=\mathcal{S}G \times \mathcal{S}G$. Any $\mathcal{S}G$-submodule $\mathfrak{C}$ of $\left(\mathcal{S}G \right)^2$ is called a $2$-quasi-abelian code over $\mathcal{S}$. For any $(a,b)\in \left(\mathcal{S}G \right)^2$, we denote $2$-quasi-abelian code $\mathfrak{C}_{c,d}$ as $\mathfrak{C}_{c,d}=\{(uc,ud) : u \in \mathcal{S}G\}$. In addition, a $2$-quasi-abelian code $\mathfrak{C}_{c,d}$ is said to be a Hermitian LCD if $\mathfrak{C}_{c,d}\cap\mathfrak{C}_{c,d}^{\perp_H}=\{0\}$. Here, $\mathfrak{C}_{c,d}^{\perp_H}=\{(x,y)\in\left(\mathcal{S}G \right)^2 : \langle (x,y), (uc,ud)\rangle=\langle x,uc \rangle+\langle y,ud\rangle=\phi(x\widehat{\sigma}{(uc)})+\phi(y\widehat{\sigma}{(ud)})=\phi(\widehat{\sigma}{(u)}(x\widehat{\sigma}{(c)}+y\widehat{\sigma}{(d)}))=0 ~\forall~ u(c,d)\in \mathfrak{C}_{c,d}\}$, where $\phi : \mathcal{S}G\rightarrow \mathcal{S}$ as $\sum\limits_{g\in G}a_gg\mapsto a_{g_0}$. We also recall the natural subjective ring homomorphism $\mathcal{S}$ to $\mathbb{F}$, i.e., $\pi : \mathcal{S}\rightarrow\mathbb{F}$ by $r\mapsto \pi(r)=r+\textbf{m}$, for any $r\in \mathcal{S}$. Naturally, we extend this map $\pi : \mathcal{S}G\rightarrow\mathbb{F}G$ by $r\mapsto \pi(r)=r+\texttt{J}(\mathcal{S}G)$, for any $r\in \mathcal{S}G$, where $\texttt{J}(\mathcal{S}G)$ is the Jacobson radical of $\mathcal{S}G$.
\begin{defn}\label{def-3.1}
 An element $r$ of $\mathcal{S}G$ is called unit element of $\mathcal{S}G$ if $\pi(r)$ is unit element of $\mathbb{F}G$.
\end{defn}
Next, we make the following proposition, which we found from \cite{McDonal1974}.
\begin{prop}\label{p-4.1}
Let $M$ be a finitely generated right $\mathcal{S}$-module with Jacobson radical $\texttt{J}(\mathcal{S})$ of $\mathcal{S}$. If $M\texttt{J}(\mathcal{S})=M,$ then $M=0.$
\end{prop}
\begin{lemma}\label{lm-4.1}
Let $\mathfrak{C}_{c,d}$ be a $2$-quasi-abelian code over $\mathcal{S}$. If $\pi(\mathfrak{C}_{c,d})$ is Hermitian LCD, then so is $\mathfrak{C}_{c,d}$.
\end{lemma}
\begin{proof} To prove of the result,
Suppose $x \in \mathfrak{C}_{c,d} \cap \mathfrak{C}_{c,d}^{\perp_H}$. Then $\pi(x) \in \pi(\mathfrak{C}_{c,d}) \cap \pi(\mathfrak{C}_{c,d}^{\perp_H})$. Since $\pi(\mathfrak{C}_{c,d}^{\perp_H}) \subseteq \pi(\mathfrak{C}_{c,d})^{\perp_H}$ and $\pi(\mathfrak{C}_{c,d})$ is Hermitian LCD by assumption, it follows that $\pi(x) = 0$, i.e., $x \in \ker(\pi) = \mathfrak{C}_{c,d} \cap \mathfrak{C}_{c,d}^{\perp_H} \cdot \texttt{J}(\mathcal{S}G)$, where $\texttt{J}(\mathcal{S}G)$ is the Jacobson radical of $\mathcal{S}G$. Hence,
$
\mathfrak{C}_{c,d} \cap \mathfrak{C}_{c,d}^{\perp_H} = \left(\mathfrak{C}_{c,d} \cap \mathfrak{C}_{c,d}^{\perp_H} \right) \cdot \texttt{J}(\mathcal{S}G).
$
Since this intersection is a finitely generated $\mathcal{S}G$-module, Proposition \ref{p-4.1} implies it must be zero. Thus, $\mathfrak{C}_{c,d}$ is Hermitian LCD. 
\end{proof}
In addition, we provide the following lemma.
\begin{lemma}\label{lm-4.1}
Let $\mathcal{S}G$ be a finite commutative ring with Jacobson radical $\texttt{J}(\mathcal{S}G)$. Then there exists a nonzero element $m \in \texttt{J}(\mathcal{S}G)$ such that $\alpha m = 0$ for all $\alpha \in \texttt{J}(\mathcal{S}G)$.
\end{lemma}
\begin{proof}
Since $\mathcal{S}G$ is finite and commutative, its Jacobson radical $\texttt{J}(\mathcal{S}G)$ is a nilpotent ideal. Let $g$ be the smallest positive integer such that $\texttt{J}(\mathcal{S}G)^g = 0$, but $\texttt{J}(\mathcal{S}G)^{g-1} \neq 0$. Then $\texttt{J}(\mathcal{S}G)^{g-1} \subseteq \texttt{J}(\mathcal{S}G)$, and for any $m \in \texttt{J}(\mathcal{S}G)^{g-1}$ and $\alpha \in \texttt{J}(\mathcal{S}G)$, we have $\alpha m \in \texttt{J}(\mathcal{S}G)^g = 0$. Hence, any nonzero $m \in \texttt{J}(\mathcal{S}G)^{g-1}$ satisfies $\alpha m = 0$ for all $\alpha \in \texttt{J}(\mathcal{S}G)$.
\end{proof}

\begin{thm}\label{th-4.0}
Let $\mathfrak{C}_{c,d}$ be a $2$-quasi-abelian code over $\mathcal{S}$ such that $\pi(\mathfrak{C}_{c,d}^\perp) = \pi(\mathfrak{C}_{c,d})^\perp$. Then $\mathfrak{C}_{c,d}$ is Hermitian LCD over $\mathcal{S}$ if and only if $\pi(\mathfrak{C}_{c,d})$ is Hermitian LCD over $\mathbb{F}$.
\end{thm}

\begin{proof}
By Lemma~\ref{lm-4.1}, if $\pi(\mathfrak{C}_{c,d})$ is Hermitian LCD over $\mathbb{F}$, then so is $\mathfrak{C}_{c,d}$.

Conversely, assume $\mathfrak{C}_{c,d} \cap \mathfrak{C}_{c,d}^{\perp\_H} = {0}$. Let $x \in \pi(\mathfrak{C}_{c,d}) \cap \pi(\mathfrak{C}_{c,d}^{\perp\_H})$. Then there exist $c \in \mathfrak{C}_{c,d}$ and $d \in \mathfrak{C}_{c,d}^{\perp\_H}$ such that $\pi(c) = \pi(d) = x$, implying $c - d \in (\mathcal{S}G)^2 \texttt{J}(\mathcal{S}G)$.
By Lemma~\ref{lm-4.1}, there exists a nonzero $m \in \texttt{J}(\mathcal{S}G)$ such that $(c - d)m = 0$, hence $cm = dm \in \mathfrak{C}_{c,d} \cap \mathfrak{C}_{c,d}^{\perp\_H} = {0}$. Thus, $cm = 0\$, which implies \$c \in (\mathcal{S}G)^2 \texttt{J}(\mathcal{S}G)$.
If $c \notin (\mathcal{S}G)^2 \texttt{J}(\mathcal{S}G)$, then $cm \neq 0$, leading to a contradiction. Therefore, $x = \pi(c) = 0$, and we conclude that
$
\pi(\mathfrak{C}_{c,d}) \cap \pi(\mathfrak{C}_{c,d}^{\perp_H}) = \{0\}.
$
Hence, $\pi(\mathfrak{C}_{c,d})$ is Hermitian LCD over $\mathbb{F}$.
\end{proof}

Now, we provide a characterization of the $2$-quasi-abelian as a Hermitian LCD. 
\begin{thm}\label{th-4.1}
Let $(c,d) \in (\mathcal{S}G)^2$. The $2$-quasi-abelian code $\mathfrak{C}_{c,d}$ over $\mathcal{S}$ is Hermitian LCD if $\pi(c\widehat{\sigma}(c) + d\widehat{\sigma}(d))$ is a unit in $\mathbb{F}G$. 
\end{thm}
\begin{proof}
Suppose $u(c,d) \in \mathfrak{C}_{c,d} \cap \mathfrak{C}_{c,d}^\perp$. Then, for all $u'(c,d) \in \mathcal{R}G$, we have
$
\langle u(c,d), u'(c,d) \rangle = \phi(u \bar{u}' (c\bar{c} + d\bar{d})) = 0.
$
Since $\pi(c\widehat{\sigma}(c) + d\widehat{\sigma}(d))$ is a unit, this implies $\phi(u\bar{u}') = 0$ for all $u' \in \mathcal{S}G$, hence $u = 0$.
Therefore, $\mathfrak{C}_{c,d}$ is Hermitian LCD.
\end{proof}
To investigate the Hermitian LCD asymptotic property for a class of $2$-quasi-abelian codes over finite chain rings, we consider a specific family of codes denoted by code $\mathfrak{C}_{1,d}$ in $\mathcal{C}_{a,b}$.
Recall that an $\mathcal{S}$-module $\mathfrak{C}_{c,d}$ is said to be free if it is isomorphic to $\mathcal{S}^t$ for some $t > 0$.
\begin{lemma}\label{lm-4.2}
$\mathfrak{C}_{1,d}$ is a free $\mathcal{S}G$ module.
\end{lemma}
\begin{proof}
Define the map $\varphi: \mathfrak{C}_{1,d} \to \mathcal{S}G$ by $\varphi(u(1,d)) = u$ for all $u \in \mathcal{S}G$. This map is clearly an $\mathcal{S}G$-module isomorphism. Hence, $\mathfrak{C}_{1,d} \cong \mathcal{S}G$ as $\mathcal{S}G$-modules, and thus $\mathfrak{C}_{1,d}$ is free.
\end{proof}

\begin{thm}\label{th-4.2}
 Let $\mathfrak{C}_{1,d}$ be a $2$-quasi-abelian code over $\mathcal{S}$. Then $\mathfrak{C}_{1,d}$ is a Hermitian LCD if and only if $\pi(\mathfrak{C}_{1,d})$ is a Hermitian LCD over $\mathbb{F}$. 
\end{thm}
\begin{proof}
 By applying Lemma~\ref{lm-4.2}, we get $\mathfrak{C}_{1,d}$ is a free $\mathcal{R}G$ module. It is well known that for $\mathfrak{C}_{1,d}$ is a free $\mathcal{S}G$ module, $\mathfrak{C}_{1,d}^{\perp_H}$ is also free $\mathcal{S}G$ module. Therefore, $\pi\left(\mathfrak{C}_{1,d}^{\perp_H}\right)=\pi\left(\mathfrak{C}_{1,d}\right)^{\perp_H}$. From this and applying Theorem~\ref{th-4.0}, the desired result follows immediately.
\end{proof}
\begin{cor}
If $d\in\mathcal{R}G$ satisfies that $\pi\left(1+d\widehat{\sigma}(d)\right)$ is aunit in $\mathbb{F}G\setminus\{1\}$, then $\mathfrak{C}_{1,d}$ is a Hermitian LCD with rate $\frac{1}{2}$.
\end{cor} 
\begin{proof}
By applying Theorem~\ref{th-4.1}, we deduce that $\mathfrak{C}_{1,d}$ is a Hermitian LCD. For the other part, by applying Lemma~\ref{lm-4.2}, it is easy to see that $\rk_{\mathcal{S}}\mathfrak{C}_{1,d}=\rk_{\mathcal{S}}\mathcal{S}G$. Hence, the rate is equal to $\frac{1}{2}$.
\end{proof}
In the following theorem, we provide a method to construct a $2$-quasi-abelian code over a finite chain ring $\mathcal{S}$ from $2$-quasi-abelian code over the residue field $\mathbb{F}$ of the chain ring $\mathcal{S}$.
\begin{thm}\label{th-4.3}
Let $\mathcal{C}_{1,\beta}=[2n,k,d]$ be a $2$-quasi-abelian code over $\mathbb{F}$. Then there exists $d\in\mathcal{S}G$ such that $\mathfrak{C}_{1,d}=[2n,k,d]$ is a $2$-quasi-abelian code over $\mathcal{S}$ with $\pi(d)=\beta$ and $\pi\left(\mathfrak{C}_{1,d}\right)=\mathcal{C}_{1,\beta}$.
\end{thm}
\begin{proof}
To prove this, we define a map $\iota : (\mathbb{F}G)^2\rightarrow(\mathcal{S}G)^2$ by $(a,b)\mapsto (a,b)$, for all $a,b\in\mathbb{F}G$, is an inclusion mapping. Under this mapping $\mathcal{C}_{1,\beta}$ is embedding in $(\mathcal{S}G)^2$, i.e., $\iota\left(\mathcal{C}_{1,\beta}\right)$ is an $\mathcal{S}G$-module of $(\mathcal{S}G)^2$. On the other hand, $\pi : \mathcal{S}G\rightarrow \mathbb{F}G$ is a surjective ring homomorphism, then for $\beta\in \mathbb{F}G$, there exists $d\in\mathcal{S}G$ such that $\pi(d)=\beta$ and $\pi\left(\mathfrak{C}_{1,d}\right)=\mathcal{C}_{1,\beta}$. By applying Lemma~\ref{lm-4.2}, we obtain that $\iota\left(\mathcal{C}_{1,\beta}\right)$ and $\mathfrak{C}_{1,d}$ are both free $\mathcal{S}G$-module. This gives $\mathfrak{C}_{1,d}=[2n,k,d]$ is a $2$-quasi-abelian code over $\mathcal{S}$.
\end{proof}
\begin{lemma}\cite{mil01}\label{le-4.3}
If $I$ is an ideal in $\mathcal{S}$ and $G$ is a finite group, then $\dfrac{\mathcal{S}G}{IG}\cong\dfrac{\mathcal{S}}{I}G$.
\end{lemma}
\begin{lemma}\label{le-4.4}
If $f$ is an idempotent element in $\mathbb{F}G$, then there is an element $f'$ in $\mathcal{S}G$ such that $\pi(f')=f$.
\end{lemma}
\begin{proof}
Since $\textbf{m}$ is the maximal ideal of $\mathcal{S}$ and it is both sided ideal, by applying Lemma \ref{le-4.3}, we get $$\dfrac{\mathcal{S}G}{\textbf{m} G}\cong\dfrac{\mathcal{S}}{\textbf{m}}G\cong\mathbb{F}G$$ and $\dfrac{\mathcal{S}G}{\textbf{m}^i G}\cong\dfrac{\mathcal{S}}{\textbf{m}^i}G$, where $1\leq i \leq s$, and $\textbf{m}^s=0$ but $\textbf{m}^{s-1}\neq 0$. 

We will prove this by induction. Inductively, we define the idempotent $f_i$ as follows
            $$f_i\in\dfrac{\mathcal{S}G}{\textbf{m}^iG}.$$
If $i=1$, in this case, we set $f_1=f$, then this lemma holds.

If $i>1$, let $f_{i-1}$ be an idempotent element in $\dfrac{\mathcal{S}G}{\textbf{m}^{i-1}G}$. Consider $h$ an element in $\dfrac{\mathcal{S}G}{\textbf{m}^iG}$ with its image $f_{i-1}$ in $\dfrac{\mathcal{S}G}{\textbf{m}^{i-1}G}$. Then it is easy to verify that $h-h^2$ in $\dfrac{\textbf{m}^{i-1}G}{\textbf{m}^iG}$. Hence, $(h-h^2)^2=0$.

Now, we set $f_i=(3h^2-2h^3)$.
Further,\begin{align*}
 f_i^2-f_i& = (3h^2-2h^3)(3h^2-2h^3-1)=h^2(3-2h)(3h^2-2h^3-1)& \\
 & =(2h-3)(1+2h)(h-h^2)^2=0.&
\end{align*}
This lemma holds for all $i$. 
Since $\textbf{m}^s=0$, so $f'=f_s$.
\end{proof}
By Equation~\eqref{sg-1}, 
\begin{eqnarray*}
\mathcal{S}G=&\mathcal{S}Ge'_0\bigoplus\mathcal{S}Ge'_1\bigoplus\cdots \bigoplus\mathcal{S}Ge_{t}, & ~\text{it follows that}\\
\mathbb{F}G=&\mathbb{F}Ge_0\bigoplus\mathbb{F}Ge_1\bigoplus\cdots\bigoplus\mathbb{F}Ge_{t},
\end{eqnarray*} and conversely by Lemma~\ref{le-4.4} and \cite[Theorem 10.2]{Hut1982}.
From this and Equation~\eqref{fg-2}, we conclude that
\begin{eqnarray*}
\mathcal{S}G=&\mathcal{S}Ge_0\bigoplus\left(\bigoplus\limits_{i=1}^{r}\mathcal{S}Ge_i\right)\bigoplus\left(\bigoplus\limits_{j=1}^{s}\mathcal{S}G\widehat{e}_{r+j}\right), &~\text{if and only if} \\
\mathbb{F}G=&\mathbb{F}Ge_0\bigoplus\left(\bigoplus\limits_{i=1}^{r}\mathbb{F}Ge_i\right)\bigoplus\left(\bigoplus\limits_{j=1}^{s}\mathbb{F}G\widehat{e}_{r+j}\right)
\end{eqnarray*}
  
Next,  we say that the sequence of Hrmitian LCD codes $\mathcal{C}_1$, $\mathcal{C}_2$, $\ldots$ of codes is asymptotically good over $\mathbb{F}$. 
Furthermore, by applying Theorems~\ref{th-4.2} and \ref{th-4.3}, there exist Hermitian LCD $2$-quasi-abelian codes $\mathfrak{C}_i$ of $(\mathcal{S}G)^2$, $i=1,2,\ldots$, such that the code sequence $\mathfrak{C}_1$, $\mathfrak{C}_2$, $\ldots$ satisfy the following conditions:
\begin{itemize}
\item[i)] The code $\mathfrak{C}_i$ of length $2n_i\rightarrow\infty$;
 \item[ii)] The rate $R(\mathfrak{C}_i)=R(\mathcal{C}_i)=\frac{1}{2}$;
 \item[iii)] The relative minimum distance $\vartriangle(\mathfrak{C}_i)=\vartriangle(\mathcal{C}_i)>\delta$ for all $i=1,2,\ldots$,
\end{itemize}  
Thus, from the above discussion, we present the class of Hermitian LCD $2$-quasi-abelian codes over finite chain rings, which is asymptotically good in
the following theorem.
\begin{thm}\label{th-4.12}
Let $n_1, n_2, \ldots$ be positive odd integers coprime to $q$ and let $G_i$ be an abelian group of finite order $n_i$ for $i=1,2,\ldots$. Then there exist Hermitian LCD $2$-quasi-abelian codes $\mathfrak{C}_i$ of $(\mathcal{S}G)^2$, $i=1,2,\ldots$, such that the sequence of codes $\mathfrak{C}_1$, $\mathfrak{C}_2$, $\ldots$ are asymptotically good.
 \end{thm}
 \begin{proof}
By applying Theorems~\ref{th-4.2} and \ref{th-4.3}, we have $\mathfrak{C}_i$ Hermitian LCD $2$-quasi-abelian codes over finite chain rings. From this and by applying Theorem~\ref{th-3.17}, the statement follows immediately.
 \end{proof}
\section*{Acknowledgement} The authors sincerely thank the Associate Editor and the referees for their meticulous reviews and insightful suggestions, which have greatly improved the quality of this paper. We also gratefully acknowledge the support of the Indian Institute of Technology Guwahati.

\end{document}